\documentclass{article}
\usepackage[utf8]{inputenc}

\usepackage{microtype}
\usepackage{graphicx}
\usepackage{booktabs} 

\usepackage{hyperref}
\usepackage{xcolor}

\usepackage[accepted]{icml2020}

\usepackage{epstopdf}
\usepackage{subcaption}
\usepackage{amsmath}
\usepackage{amsthm}
\usepackage{amssymb}
\usepackage{bm}
\usepackage{bbm}
\usepackage{mathdots}
\usepackage{multicol}
\usepackage{multirow}
\usepackage{enumitem}
\usepackage{booktabs}
\usepackage{algorithm}
\usepackage{algorithmic}
\usepackage[backgroundcolor=gray!20, linecolor=gray, textsize=scriptsize]{todonotes}

\usepackage{forest}
\usepackage{tikz}
\usetikzlibrary{calc, graphs, graphs.standard, shapes, arrows, positioning, decorations.pathreplacing, decorations.markings, decorations.pathmorphing, fit, matrix, patterns, shapes.misc}

\tikzset{cross/.style={cross out, draw,
			minimum size=2*(#1-\pgflinewidth),
			inner sep=0pt, outer sep=0pt}}

\tikzstyle{vecArrow} = [thick, decoration={markings,mark=at position
		1 with {\arrow[semithick]{open triangle 60}}},
double distance=1.4pt, shorten >= 5.5pt,
preaction = {decorate},
postaction = {draw,line width=1.4pt, white,shorten >= 4.5pt}]
\tikzstyle{innerWhite} = [semithick, white,line width=1.4pt, shorten >= 4.5pt]

\usepackage[capitalize, nameinlink]{cleveref}

\usepackage{thmtools}
\usepackage{thm-restate}

\newtheorem{theorem}{Theorem}
\newtheorem{lemma}[theorem]{Lemma}
\newtheorem{corollary}[theorem]{Corollary}
\newtheorem{definition}[theorem]{Definition}
\newtheorem{observation}[theorem]{Observation}
\newtheorem{proposition}[theorem]{Proposition}

\newtheorem{claim}[theorem]{Claim}

\crefname{claim}{Claim}{Claims}
\crefname{fact}{Fact}{Facts}
\crefname{example}{Example}{Examples}

\DeclareMathOperator*{\argmin}{arg\,min}
\newcommand{\E}{\textrm{\textbf{E}}}
\renewcommand{\P}{\textrm{\textbf{P}}}
\newcommand{\eps}{\varepsilon}
\newcommand{\e}{\textrm{e}}

\newcommand{\reassign}{\texttt{reassign}}
\newcommand{\kmp}{\texttt{$k$-means++}}

\newcommand{\localsearch}{\texttt{LocalSearch++}}
\renewcommand{\O}{\mathcal{O}}
\newcommand{\Qin}{Q_{\textrm{in}}}

\allowdisplaybreaks

\emergencystretch=\maxdimen
\icmltitlerunning{\kmp: constant approximation}

\begin{document}

\twocolumn[
	\icmltitle{\kmp: few more steps yield constant approximation}



	\icmlsetsymbol{equal}{*}

	\begin{icmlauthorlist}
		\icmlauthor{Davin Choo}{equal,eth}
		\icmlauthor{Christoph Grunau}{equal,eth}
		\icmlauthor{Julian Portmann}{equal,eth}
		\icmlauthor{Václav Rozhoň}{equal,eth}
	\end{icmlauthorlist}

	\icmlaffiliation{eth}{ETH Zürich}

	\icmlcorrespondingauthor{Davin Choo}{chood@ethz.ch}
	\icmlcorrespondingauthor{Christoph Grunau}{cgrunau@ethz.ch}
	\icmlcorrespondingauthor{Julian Portmann}{pjulian@ethz.ch}
	\icmlcorrespondingauthor{Václav Rozhoň}{rozhonv@ethz.ch}

	\icmlkeywords{k-means, clustering, local search, theory}

	\vskip 0.3in
]

\printAffiliations{\icmlEqualContribution}

\begin{abstract}
	The \kmp{} algorithm of Arthur and Vassilvitskii (SODA 2007) is a state-of-the-art algorithm for solving the $k$-means clustering problem and is known to give an $\O(\log k)$-approximation in expectation.
	Recently, Lattanzi and Sohler (ICML 2019) proposed augmenting \kmp{} with $\O(k \log \log k)$ local search steps to yield a constant approximation (in expectation) to the $k$-means clustering problem.
	In this paper, we improve their analysis to show that, for any arbitrarily small constant $\eps > 0$, with only $\eps k$ additional local search steps, one can achieve a constant approximation guarantee (with high probability in $k$), resolving an open problem in their paper.
\end{abstract}

\section{Introduction}

$k$-means clustering is an important unsupervised learning task often used to analyze datasets.
Given a set $P$ of points in $d$-dimensional Euclidean space $\mathbb{R}^d$ and an integer $k$, the task is to partition $P$ into $k$ clusters while minimizing the total cost of the partition.
Formally, the goal is to find a set $C \in \mathbb{R}^d$ of $k$ centers minimizing the following objective:
\[
	\sum_{p \in P} \min_{c \in C} \Vert p - c \Vert^2,
\]
where points $p \in P$ are assigned to the closest candidate center $c \in C$.

Finding an optimal solution to this objective was proven to be NP-hard \cite{aloise2009np, mahajan2009planar}, and the problem was even shown to be hard to approximate to arbitrary precision \cite{awasthi2015hardness, lee2017improved}.
The currently best known approximation ratio is 6.357 \cite{ahmadian2019better}, while other constant factor approximation algorithms exist \cite{jain2001approximation, kanungo2004local}.
For constant dimensions $d$, $(1 + \eps)$-approximation algorithms are known \cite{cohen2018fast, cohen2019local, friggstad2019local, bandyapadhyay2015variants}.
However, these algorithms are mainly of theoretical interest and not known to be efficient in practice.

On the practical side of things, the canonical $k$-means algorithm \cite{lloyd1982least} proved to be a good heuristic.
Starting with $k$ initial points (e.g. chosen at random), Lloyd's algorithm iteratively, in an alternating minimization manner, assigns points to the nearest center and updates the centers to be the centroids of each cluster, until convergence.
Although the alternating minimization provides no provable approximation guarantee, Lloyd's algorithm never increases the cost of the initial clustering.
Thus, one way to obtain theoretical guarantees is to provide Lloyd's algorithm with a provably good initialization.

The \kmp{} algorithm (see \cref{alg:kmp}) of Arthur and Vassilvitskii \yrcite{arthur2007k} is a well-known algorithm for computing an initial set of $k$ centers with provable approximation guarantees.
The initialization is performed by incrementally choosing $k$ initial seeds for Lloyd using $D^2$-sampling, i.e., sample a point with probability proportional to its squared distance to the closest existing center.
They showed that the resultant clustering is an $\O(\log k)$-approximation in expectation.
This theoretical guarantee is substantiated by empirical results showing that \kmp{} can heavily outperform random initialization, with only a small amount of additional computation time on top of running Lloyd's algorithm.
However, lower bound analyses \cite{brunsch2013bad, bhattacharya2016tight} show that there exist inputs where \kmp{} is $\Omega(\log k)$-competitive with high probability in $k$.

Recently, Lattanzi and Sohler \yrcite{lattanzi2019better} proposed a variant of local search after picking $k$ initial centers via \kmp{} (see \cref{alg:localsearch}):
In each step, a new point is sampled with probability proportional to its current cost and used to replace an existing center such as to maximize the cost reduction.
If all possible swaps increase the objective cost, the new sampled point is discarded.
Following their notation, we refer to this local search procedure as \localsearch{}.
They showed that performing $\O(k \log \log k)$ steps of \localsearch{} after \kmp{} improves the expected approximation factor from $\O(\log k)$ to $\O(1)$, and stated that it is an interesting open question to prove that $\O(k)$ local search steps suffice to obtain a constant factor approximation.

\begin{algorithm}[tb]
	\caption{\kmp{} seeding}
	\label{alg:kmp}
	Input: $P$, $k$, $\ell$
	\begin{algorithmic}[1]
		\STATE Uniformly sample $p \in P$ and set $C = \{ p \}$.
		\FOR{$i \leftarrow 2, 3, \dots, k$}
		\STATE Sample $p \in P$ w.p. $\frac{cost(p, C)}{\sum_{q \in P} cost(q, C)}$ and add it to $C$.
		\ENDFOR
	\end{algorithmic}
\end{algorithm}

\begin{algorithm}[tb]
	\caption{One step of \localsearch}
	\label{alg:localsearch}
	{\bfseries Input:} $P$, $C$
	\begin{algorithmic}[1]
		\STATE Sample $p \in P$ with probability $\frac{cost(p, C)}{\sum_{q \in P} cost(q, C)}$
		\STATE $p' = \argmin_{q \in C} cost(P, C \setminus \{ q \} \cup \{ p \})$
		\IF{$cost(P, C \setminus \{ p' \} \cup \{ p \} ) < cost(P, C)$}
		\STATE $C = C \setminus \{ p' \} \cup \{ p \}$
		\ENDIF
		\STATE \textbf{return} $C$
	\end{algorithmic}
\end{algorithm}

\subsection{Our contribution}

In this paper, we answer the open question by Lattanzi and Sohler \yrcite{lattanzi2019better} in the affirmative.
We refine their analysis to show that with only $\eps k$ additional local search steps, one can achieve an approximation guarantee of $\O(1 / \eps^3)$ with probability $1 - \exp \left( -\Omega \left( k^{0.1} \right) \right)$.
Compared to \cite{lattanzi2019better}, we improve the number of search steps needed to achieve a constant approximation from $\O(k \log \log k)$ to just $\eps k$.
Furthermore, our statement holds with high probability in $k$, while their statement gives a guarantee in expectation.

\begin{restatable}[Main theorem]{theorem}{mainthm}
	\label{thm:main}
	Let $k \in \Omega(1/ \eps^{20})$ and $0 < \eps \leq 1$.
	Suppose we run \cref{alg:kmp} followed by $\ell = \eps k$ steps of \cref{alg:localsearch}.
	We have $cost(P,C) \leq \left( 10^{30} / \eps^3 \right) \cdot cost(P,C^*)$ with probability at least $1 - \exp(-\Omega(k^{0.1}))$.
\end{restatable}

\subsection{Related Work}

Another variant of local search was analyzed by Kanungo et al. \yrcite{kanungo2004local}:
In each step, try to improve by swapping an existing center with an input point.
Although they showed that this eventually yields a constant approximation, the number of required steps can be very large.

Under the bicriteria optimization setting, Aggarwal et al. \yrcite{aggarwal2009adaptive} and Wei \yrcite{wei2016constant} proved that if one over-samples and runs $D^2$-sampling for $\O(k)$ steps (instead of just $k$), one can get a constant approximation of the $k$-means objective with these $\O(k)$ centers.
We note that a single step of \localsearch{} has almost the same asymptotic running time as over-sampling once using $D^2$-sampling (see \cref{para:runningtime} for details) while enforcing the constraint of \emph{exactly} $k$ centers.
However, their results are stronger in terms of approximation guarantees:
Aggarwal et al. \yrcite{aggarwal2009adaptive} proved that in $\O(k)$ steps one achieves a $4+\eps$ approximation to the optimal cost with constant probability, while Wei \yrcite{wei2016constant} proved that after $\eps k$ more sampling steps one achieves an $\O(1/\eps)$-approximation in expectation.

Other related work include speeding up \kmp{} via approximate sampling \cite{bachem2016fast}, approximating \kmp{} in the streaming model \cite{ackermann2012streamkm++}, and running \kmp{} in a distributed setting \cite{bahmani2012scalable}.

\subsection{Our method, in a nutshell}
Lattanzi and Sohler showed that given any clustering with approximation ratio of at least 500, a single step of \localsearch{} improves the cost by a factor of $1 - 1/(100k)$, with constant probability.
In general, one cannot hope to asymptotically improve their bound\footnote{Consider a $(k-1)$-dimensional simplex with $n/k$ points at each corner and a clustering with all $k$ centers in the same corner.
	Swapping one center to another corner of the simplex improves the solution by a factor of $1 - \Theta(1/k)$. However, we do not expect to get such a solution after running the \kmp{} algorithm.}.
Instead, our improvement comes from structural insights on solutions provided by \cref{alg:kmp} and \cref{alg:localsearch}.
We argue that if we have a bad approximation at any point in time, the next step of \cref{alg:localsearch} drastically improves the cost with a positive constant probability.

To be more specific, consider an optimal clustering $OPT$.
We prove that throughout the course of \cref{alg:localsearch}, most centroids of $OPT$ clusters have a candidate center in the current solution close to it.
Through a chain of technical lemmas, \`{a} la \cite{lattanzi2019better}, we conclude that the new sampled point is close to the centroid of one of the (very few) costly $OPT$ clusters with constant probability.
Then, we argue that there is an existing candidate that can be swapped with the newly sampled point to improve the solution quality substantially.
Putting everything together, we get that the solution improves by a factor of $1 - \Theta(\sqrt[3]{\alpha} /k)$ with constant probability in a single \localsearch{} step, where $\alpha$ is the approximation factor of the current solution.
This improved multiplicative cost reduction suffices to prove our main result.

In \cref{sec:preliminaries}, we introduce notation and crucial definitions, along with several helpful lemmas.
In \cref{sec:walk_through_proof}, we walk the reader through our proof while deferring some lengthy proofs to the supplementary material.

\section{Preliminaries}
\label{sec:preliminaries}

Let $P$ be a set of points in $\mathbb{R}^d$.
For two points $p, q \in \mathbb{R}^d$, let $\Vert p - q \Vert$ be their Euclidean distance.
We denote $C \subseteq P$ as the set of candidate centers and $C^* = OPT$ as the centers of a (fixed) optimal solution, where $|C^*| = k$.
Note that a center $c^* \in C^*$ may \emph{not} be a point from $P$ while all candidates $c \in C$ are actual points from $P$.
For $c^* \in C^*$, the set $Q_{c^*}$ denotes the points in $P$ that $OPT$ assigns to $c^*$.
We define $cost(P,C) = \sum_{p \in P} \min_{c \in C} \Vert p - c \Vert^2$ as the cost of centers $C$, where $cost(P,C^*)$ is the cost of an optimal solution.
When clear from context, we also refer to the optimal cost as $OPT$.
For an arbitrary set of points $Q$, we denote their centroid by $\mu_Q = (1 / |Q|) \cdot \sum_{q \in Q}q$.
Note that $\mu_Q$ may not be a point from $Q$.

For the sake of readability, we drop the subscript $Q$ when there is only one set of points in discussion and we drop braces when describing singleton sets in $cost(\cdot, \cdot)$.
We will also ignore rounding issues as they do not play a critical role asymptotically.

We now define $D^2$-sampling introduced in \kmp{}.

\begin{definition}[$D^2$-sampling]
	\label{def:d2-sampling}
	Given a set $C \subseteq P$ of candidate centers, we sample a point $p \in P$ with probability $\P[p] = cost(p,C) / \sum_{p \in P} cost(p,C)$.
\end{definition}

The following folklore lemma describes an important property of the cost function.
This is analogous to the bias-variance decomposition in machine learning and to the parallel axis theorem in physics \cite{aggarwal2009adaptive}.
As the variable naming suggests, we will use it with $Q$ being an $OPT$ center and $c$ being a candidate center.

\begin{lemma}
	\label{prop:physics_lemma}
	Let $Q \subseteq P$ be a set of points.
	For any point $c \in P$ (possibly not in $Q$),
	\[
		cost(Q, c) = |Q| \cdot \Vert c - \mu_Q \Vert^2 + cost(Q, \mu_Q)
	\]
\end{lemma}

To have a finer understanding of the cluster structure, we define the notions of \emph{settled} and \emph{approximate} clusters.
Consider an arbitrary set of points $Q \subseteq P$ (e.g. some cluster of $OPT$).
We define
\[
	{R_{Q,\beta} = \{q \in Q: \Vert q - \mu_Q \Vert^2 \leq (\beta / |Q| ) \cdot cost(Q, \mu_Q)\}}
\]
as the subset of points in $Q$ that are within a certain radius from $\mu_Q$ (i.e. ``close'' with respect to $\beta$).
As $\beta$ decreases, the condition becomes stricter and the set $R_{Q,\beta}$ shrinks.

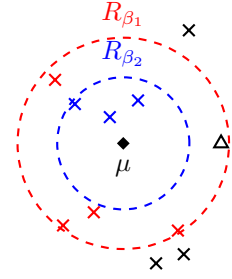
\begin{figure}[t]
	\centering
	\begin{tikzpicture}
		\draw[thick] (-0.644, 0.523) node[cross=3pt, blue] {};
		\draw[thick] (-0.801, -1.092) node[cross=3pt, red] {};
		\draw[thick] (0.436, -1.596) node[cross=3pt, black] {};
		\draw[thick] (-0.392, -0.917) node[cross=3pt, red] {};
		\draw[thick] (0.196, 0.571) node[cross=3pt, blue] {};
		\draw[thick] (0.804, -1.475) node[cross=3pt, black] {};
		\draw[thick] (0.872, 1.504) node[cross=3pt, black] {};
		\draw[thick] (-0.907, 0.845) node[cross=3pt, red] {};
		\draw[thick] (-0.179, 0.349) node[cross=3pt, blue] {};
		\draw[thick] (0.726, -1.16) node[cross=3pt, red] {};

		\node[draw, thick, regular polygon,regular polygon sides=3, minimum size=6pt, inner sep=0pt] at (1.3,0) {};

		\draw (0,0) node[fill=black, diamond, inner sep=0pt, minimum size=5pt, label=below:\textcolor{black}{$\mu$}] {};
		\node[draw, thick, dashed, red, circle, minimum size=80pt, label=above:\textcolor{red}{$R_{\beta_1}$}] at (0,0) {};
		\node[draw, thick, dashed, blue, circle, minimum size=50pt, label=above:\textcolor{blue}{$R_{\beta_2}$}] at (0,0) {};
	\end{tikzpicture}
	\caption{
		$\times$ are points in set $Q$ with centroid $\mu$ (which may \emph{not} be a point from $Q$).
		For $\beta_1 > \beta_2$, $R_{\beta_2} \subseteq R_{\beta_1} \subseteq Q$.
		$\triangle$ represents a candidate center $c \in C$ that does \emph{not} belong to $Q$.
		Since $c \not\in Q$, cluster $Q$ is \emph{not} $\beta_1$-settled even though $\Vert c - \mu_Q \Vert^2 \leq (\beta_1 / |Q|) \cdot cost(Q, \mu_Q)$.
	}
	\label{fig:example-Q}
\end{figure}

\begin{definition}[$\beta$-settled]
	\label{def:settled}
	An $OPT$ cluster $Q$ is $\beta$-settled if $R_{\beta} \cap C \neq \emptyset$.
	That is, there is a candidate center $c \in Q$ with distance at most $(\beta / |Q|) \cdot cost(Q, \mu_Q)$ from $\mu_Q$.
\end{definition}

\begin{definition}[$\alpha$-approximate]
	\label{def:approximate}
	An $OPT$ cluster $Q$ is $\alpha$-approximate if $cost(Q,C) \leq \alpha \cdot cost(Q,\mu)$.
\end{definition}

Intuitively, a $\beta$-settled cluster $Q$ has small $cost(Q,C)$.
As settled-ness requires a candidate $c \in C$ to belong to cluster $Q$, an unsettled cluster $Q$ could have small $cost(Q,C)$.
See \cref{fig:example-Q} for an illustration of $Q$, $\mu_Q$, $R_{Q,\beta}$ and $\beta$-settled.
We now relate the definitions of settled and approximate.

\begin{lemma}
	\label{prop:beta-settled}
	Suppose $Q$ is a cluster of $OPT$ that is $\beta$-settled.
	Then, $cost(Q,C) \le (\beta+1) \cdot cost(Q, \mu)$.
	In other words, $\beta$-settled implies $(\beta+1)$-approximate.
\end{lemma}
\begin{proof}
	For any $\beta$-settled cluster $Q$, there is some candidate center $c \in C$ in $R_\beta$, so
	\begin{align*}
		 & cost(Q,C)
		\leq cost(Q,c)
		= |Q| \cdot \Vert c - \mu \Vert^2 + cost(Q, \mu) \\
		 & \leq (\beta+1) \cdot cost(Q,\mu)
		= (\beta+1) \cdot cost(Q,C^*)
	\end{align*}
\end{proof}

It is also useful to consider the contrapositive of \cref{prop:beta-settled}.

\begin{corollary}
	\label{cor:mod-lem6}
	Let $Q$ be a cluster of $OPT$.
	If $cost(Q,C) > (\beta + 1) \cdot cost(Q,\mu)$, then $\Vert c - \mu \Vert^2 > (\beta / |Q|) \cdot cost(Q,\mu)$ for any candidate center $c \in C$.
	That is, $Q$ is $\beta$-unsettled.
\end{corollary}

In our analysis, we will prove statements about clusters being settled for general values of $\beta$.
However, there are uncountably many possible $\beta$'s and therefore we cannot do a union bound over all possible choices of $\beta$.
Using a similar idea to $\eps$-net arguments, we will discretize the set of $\beta$'s into a sufficiently small finite set of \emph{legal} values.

\begin{definition}[Legal $\beta$ values]
	A parameter $\beta$ is \emph{legal} if $\beta \in \mathcal{B} = \{ 2^{i} : \, i \in \{3, 4, \dots, 0.3 \cdot \log k \} \}$.
	In particular, this implies that all legal $\beta$ are at least $8$ and at most $k^{0.3}$.
\end{definition}

\section{Analysis}
\label{sec:walk_through_proof}

In this section, we present the key ideas of our proof while deferring some
details to the supplementary material.

Following the proof outline of Lattanzi and Sohler \yrcite{lattanzi2019better}, we perform a more intricate analysis of \localsearch{}.
Their key lemma shows that, with constant probability, the cost of the solution decreases by a factor of $1 - 1/(100k)$ after one local search step.

\paragraph{Lemma 3 in \cite{lattanzi2019better}}
Let $P$ be a set of points and $C$ be a set of centers with $cost(P,C) > 500\: OPT$.
Denote the updated centers by $C' = \localsearch(P,C)$.
Then, with probability $1/1000$, $cost(P,C') \le \left( 1 - 1/(100k) \right) cost(P,C)$.

The above lemma implies that we expect the cost of the current solution to drop by a constant factor after $\O(k)$ \localsearch{} steps (unless we already have a constant approximation of the optimum).
Since we start with a solution that is an $\O(\log k)$-approximation in expectation, we expect that after $\O(k \log\log k)$ iterations, the cost of our solution drops to a constant.
This yields the main theorem of \cite{lattanzi2019better}.

\paragraph{Theorem 1 in \cite{lattanzi2019better}}
Let $P$ be a set of points and $C$ be the output of \kmp{} followed by at least $100000\: k \log \log k$ many local search steps.
Then, we have $\E[cost(P, C)] \in \O(cost(P, C^*))$.
The running time of the algorithm is $\O(dnk^2 \log \log k)$.

Our improvements rely on the following structural observation:
After running \kmp{}, most of the clusters of the optimal solution are already ``well approximated'' with high probability in $k$.

\subsection{Structural analysis}

In this subsection, we study the event of sampling a point from a $\beta$-unsettled cluster and making it $\beta$-settled.
This allows us to prove concentration results about the number of $\beta$-settled clusters, which we will use in the next subsection.

Suppose $\alpha$ is the current approximation factor.
The result below states that with good probability, the new sampled point is from a cluster that is currently badly approximated.

\begin{restatable}{lemma}{lemsamplefromunsettled}
	\label{prop:sampled_point_is_from_bad_cluster}
	Suppose that $cost(P,C) = \alpha \cdot cost(P,C^*)$ and we $D^2$-sample a point $p \in P$.
	Consider some fixed $\beta \geq 1$.
	Then, with probability at least $1 - \beta / \alpha$, the sampled point $p$ is from a cluster $Q$ with $cost(Q,C) \ge \beta \cdot cost(Q,\mu_Q)$.
\end{restatable}
\begin{proof}
	Let $\widetilde{Q}$ be the union of all clusters $Q$ such that $cost(Q,C) < \beta \cdot cost(Q,\mu_Q)$.
	By definition, $D^2$-sampling will pick a point from $\widetilde{Q}$ with probability at most $cost(\widetilde{Q},C) / cost(P,C) \leq \beta / \alpha$.
\end{proof}

Similar to most work\footnote{Cf. Lemma 2 in \cite{arthur2007k}, Lemma 5 of \cite{aggarwal2009adaptive}, Lemma 6 of \cite{lattanzi2019better}.} on \kmp{}, we need a sampling lemma stating that if we sample a point within a cluster $Q$ (according to $D^2$ weights of points in $Q$), then the sampled point will be relatively close to $\mu_Q$ with good probability.

\begin{restatable}{lemma}{lemsampleunsettled}
	\label{lem:sampled_point_is_from_unsettled_cluster}
	Suppose that $Q$ is an $OPT$ cluster with $cost(Q, C) \ge \beta \cdot cost(Q, \mu_Q)$ for some $\beta \ge 4$ and we $D^2$-sample a point $p \in Q$.
	Then, with probability at least $1 - 6 / \sqrt{\beta}$, $Q$ becomes $(\beta - 1)$-settled.
	That is, $\Vert p - \mu_Q \Vert^2 \le \left( (\beta-1) / |Q| \right) \cdot cost(Q,\mu_Q)$ and $cost(Q, p) \le \beta \cdot cost(Q, \mu_Q)$.
\end{restatable}

\begin{proof}
	Let $\beta' \ge \beta$ be the \emph{exact} approximation factor.
	i.e. $cost(Q,C) = \beta' \cdot cost(Q,\mu)$.
	We define sets $\Qin$ and $P'_{\textrm{in}}$:

	\begin{align*}
		\Qin             & = \left\{ q \in Q : \Vert q - \mu_Q \Vert \leq \sqrt{\frac{\beta - 2}{|Q|}cost(Q,\mu_Q)} \right\}  \\
		P'_{\textrm{in}} & = \left\{ p \in P : \Vert p - \mu_Q \Vert \leq \sqrt{\frac{\beta' - 2}{|Q|}cost(Q,\mu_Q)} \right\}
	\end{align*}

	By definition, we have  $\Qin \subseteq Q$ and $\Qin \subseteq P'_{\textrm{in}}$.
	However, $P'_{\textrm{in}} \subseteq Q$ does not hold in general.

	\cref{prop:physics_lemma} tells us that $P'_{\textrm{in}} \cap C = \emptyset$.
	Otherwise $\beta' \cdot cost(Q, \mu_Q) = cost(Q, C) \le (\beta'-1) \cdot cost(Q, \mu_Q)$, which is a contradiction.
	Furthermore, it holds that $|Q \setminus \Qin| \le |Q|/ (\beta-2)$.
	Otherwise $cost(Q, \mu_Q) > \left( |Q| / (\beta-2) \right) \cdot \left( (\beta - 2) / |Q| \right) \cdot cost(Q,\mu_Q)$, which is a contradiction.
	Hence, $|\Qin| \ge (1 - 1 / (\beta-2)) \cdot |Q|$.

	Let $d_i = \Vert q_i - \mu_Q \Vert$ be the distance of the $i$-th point of $\Qin$ from $\mu_Q$, so $\sum_{i=1}^{|\Qin|} d_i^2 \le cost(Q, \mu_Q)$.
	By the Cauchy-Schwarz inequality, we have
	$
		\sum_{i=1}^{|\Qin|} d_i \le \sqrt{|\Qin| \cdot \sum_{i=1}^{|\Qin|} d_i^2} \le \sqrt{|Q| \cdot cost(Q,\mu_Q)}
	$.
	Since $P'_{\textrm{in}} \cap C = \emptyset$, triangle inequality tells us that $\sqrt{cost(q_i, C)}
		\ge \sqrt{cost(\mu_Q, C)} - \sqrt{cost(q_i, \mu_Q)}
		\ge \sqrt{\left( (\beta' - 2) / |Q| \right) \cdot cost(Q,\mu_Q)} - d_i$ for each point $q_i \in \Qin$.
	Thus,
	\begin{align*}
		    & \; cost(\Qin, C) = \sum_{i=1}^{|\Qin|} cost(q_i, C)                                                                          \\
		\ge & \; \sum_{i=1}^{|\Qin|} \left( \sqrt{\frac{\beta'-2}{|Q|}cost(Q,\mu_Q)} - d_i \right)^2                                       \\
		\ge & \; \sum_{i=1}^{|\Qin|} \frac{\beta'-2}{|Q|}cost(Q,\mu_Q) - 2 d_i \sqrt{\frac{\beta'-2}{|Q|}cost(Q,\mu_Q)}                    \\
		=   & \; \frac{|\Qin|}{|Q|}(\beta'-2) cost(Q,\mu_Q)                                                                                \\
		    & \; - 2 \sqrt{\frac{\beta'-2}{|Q|}cost(Q,\mu_Q)} \cdot \sum_{i=1}^{|\Qin|} d_i                                                \\
		\ge & \; \left( 1 - \frac{1}{\beta-2} \right) \left( \beta'-2 \right) cost(Q,\mu_Q)                                                \\
		    & \; - 2 \sqrt{\frac{\beta'}{|Q|}cost(Q,\mu_Q)} \cdot \sqrt{|Q|cost(Q,\mu_Q)}                                                  \\
		=   & \; \left( \left( 1 - \frac{1}{\beta-2} \right) \left( \beta'-2 \right) - 2 \sqrt{\beta'} \right) cost(Q,\mu_Q)               \\
		=   & \; \left( \left( 1 - \frac{1}{\beta-2} \right) \left( 1-\frac{2}{\beta'} \right) - \frac{2}{\sqrt{\beta'}} \right) cost(Q,C) \\
		\ge & \; \left( 1 - \frac{2+2+2}{\sqrt{\beta}} \right) cost(Q,C)                                                                   \\
		=   & \; \left( 1 - \frac{6}{\sqrt{\beta}} \right) cost(Q,C)
	\end{align*}
	Hence, the probability that the sampled point $p$ is taken from $\Qin$ is at least $1 - 6 / \sqrt{\beta}$.
	Having sampled a point $p \in Q$ with $\Vert p - \mu_Q \Vert^2 \le ((\beta - 1) / |Q|) \cdot cost(Q, \mu_Q)$, \cref{prop:physics_lemma} tells us that the cost of cluster $Q$ is at most $\beta \cdot cost(Q, \mu_Q)$.
\end{proof}

\begin{corollary}
	\label{cor:I_like_Sebastians_chocolate}
	Fix $\alpha \geq 10$ such that $cost(P,C) = \alpha \cdot OPT$ and let $1 < \beta \le \alpha^{2/3}$.
	Suppose that we $D^2$-sample a new point $p \in P$.
	Then, with probability at least $1 - 8 / \sqrt{\beta}$, the sampled point $p$ is from a $\beta$-unsettled cluster and this cluster becomes $\beta$-settled.
\end{corollary}
\begin{proof}
	\cref{prop:sampled_point_is_from_bad_cluster} tells us that, with probability at least $1 - (2 \beta) / \alpha$, we sample from an $OPT$ cluster $Q$ with $cost(Q,C) \ge (2 \beta) \cdot cost(Q,\mu_Q)$.
	As $cost(Q,C) > ( \beta + 1) \cdot cost(Q,\mu_Q)$, \cref{cor:mod-lem6} implies that $Q$ is $\beta$-unsettled.
	According to \cref{lem:sampled_point_is_from_unsettled_cluster}, $Q$ becomes $\beta$-settled with probability at least $1 - 6 / \sqrt{\beta + 1} \ge 1 - 6 / \sqrt{\beta}$.
	As $\beta \leq \alpha^{2/3}$, the probability of the first event is at least
	$1 - (2 \beta) / \alpha \ge 1 - (2 \beta) / \beta ^{3/2} = 1 - 2 / \sqrt{\beta}$.
	Thus, the joint event of sampling from a $\beta$-unsettled cluster and making it $\beta$-settled happens with probability at least $1 - 8 / \sqrt{\beta}$.
\end{proof}

We can now use \Cref{cor:I_like_Sebastians_chocolate}, together with a Chernoff Bound, to upper-bound the number of $\beta$-unsettled clusters.
First, we show that with high probability, \kmp{} leaves only a small number of clusters $\beta$-unsettled for every legal $\beta \le \alpha^{2/3}$.
Then, we show that this property is maintained throughout the course of the local search.

\begin{restatable}{lemma}{lemqualitygood}
	\label{lem:quality_deteriorates_slowly}
	After running \kmp{} (for $k$ steps) and $\ell \le k$ steps of \localsearch{}, let $C$ denote the set of candidate centers and $\alpha \geq 1$ be the approximation factor.
	Then, with probability at least $1 - \exp(-\Omega(k^{0.1}))$, there are at most $\left( 30 k \right) / \sqrt{\beta}$ clusters that are $\beta$-unsettled, for any legal $\beta \leq \alpha^{2/3}$.
\end{restatable}

Recall that every legal $\beta$ is smaller than $k^{0.3}$, as for larger $\beta$ one cannot obtain strong concentration results.
For reasonably small $\alpha$, \cref{lem:quality_deteriorates_slowly} allows us to conclude that there are at most $\O(k / \sqrt{\alpha^{2/3}}) = O(k/\sqrt[3]{\alpha})$ clusters that are $\alpha^{2/3}$-unsettled, with high probability in $k$.
Conditioned on this event, we can expect a stronger multiplicative improvement in one iteration of \localsearch{} compared to Lemma 3 of \cite{lattanzi2019better}.

\subsection{One step of \localsearch{}}

Given the structural analysis of the previous section, we can now analyze the \localsearch{} procedure.
First, we will identify clusters whose removal will not significantly increase the current cost, thus making them good candidates for swapping with the newly sampled center.

To that end, we define subsets of \emph{matched} and \emph{lonely} candidate centers $M \subseteq C$ and $L \subseteq C$.
The notion of lonely centers came from Kanungo et al. \yrcite{kanungo2004local}.
To describe the \emph{same} subsets, Lattanzi and Sohler \yrcite{lattanzi2019better} used the notation $H$ and $L$, while we use $M$ and $L$.
For an illustration of these definitions, see \cref{fig:pairing-of-centers}.

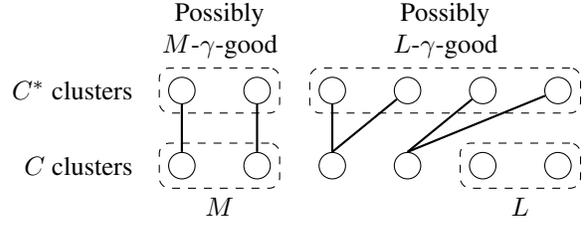
\begin{figure}[t]
	\centering
	\begin{tikzpicture}
		\node[draw, circle, minimum size=10pt] at (1,1) (q1) {};
		\node[draw, circle, minimum size=10pt] at (2,1) (q2) {};
		\node[draw, circle, minimum size=10pt] at (3,1) (q3) {};
		\node[draw, circle, minimum size=10pt] at (4,1) (q4) {};
		\node[draw, circle, minimum size=10pt] at (5,1) (q5) {};
		\node[draw, circle, minimum size=10pt] at (6,1) (q6) {};
		\node[text centered, align=right, left=10pt of q1] {$C^*$ clusters};

		\node[draw, circle, minimum size=10pt] at (1,0) (c1) {};
		\node[draw, circle, minimum size=10pt] at (2,0) (c2) {};
		\node[draw, circle, minimum size=10pt] at (3,0) (c3) {};
		\node[draw, circle, minimum size=10pt] at (4,0) (c4) {};
		\node[draw, circle, minimum size=10pt] at (5,0) (c5) {};
		\node[draw, circle, minimum size=10pt] at (6,0) (c6) {};
		\node[text centered, align=right, left=10pt of c1] {$C$ clusters};

		\draw[thick] (q1) -- (c1.north);
		\draw[thick] (q2) -- (c2.north);
		\draw[thick] (q3) -- (c3.north);
		\draw[thick] (q4) -- (c3.north);
		\draw[thick] (q5) -- (c4.north);
		\draw[thick] (q6) -- (c4.north);

		\node[draw, dashed, rounded corners, fit=(q1)(q2), label={[align=center]above:Possibly\\ $M$-$\gamma$-good}] {};
		\node[draw, dashed, rounded corners, fit=(q3)(q4)(q5)(q6), label={[align=center]above:Possibly\\ $L$-$\gamma$-good}] {};
		\node[draw, dashed, rounded corners, fit=(c1)(c2), label=below:$M$] {};
		\node[draw, dashed, rounded corners, fit=(c5)(c6), label=below:$L$] {};

	\end{tikzpicture}
	\caption{
		Let $k=6$.
		The top row represents the $k$ $OPT$ centers $C^*$.
		The bottom row represents the $k$ candidate centers $C$.
		Each $OPT$ center is connected to the closest candidate center by a line.
		Observe that $M$ and $L$ are subsets of $C$, where some candidate centers might be in neither $M$ nor $L$, and that $\gamma$-goodness is defined on the $OPT$ centers.
		This example shows a tight case for \cref{obs:double_counting_L}.
	}
	\label{fig:pairing-of-centers}
\end{figure}

\begin{definition}[$M$ and $L$ candidates]
	\label{def:HL}
	We assign $OPT$ centers $c^* \in C^*$ to candidate centers $c \in C$, and define the notion of matched ($M$) and lonely ($L$) on candidates based on assignment outcome.
	For each $c^* \in C^*$, assign $c^*$ to the closest $c \in C$, breaking ties arbitrarily.
	We say candidate $c \in C$ is \emph{matched} if there is \emph{exactly one} $c^* \in C^*$ assigned to it and we call $c^*$ the \emph{mate} of $c$.
	We say candidate $c \in C$ is \emph{lonely} if there is \emph{no} $c^* \in C^*$ assigned to it.

	We define $M \subseteq C$ as the set of matched candidates and $L \subseteq C$ as the set of lonely candidates.
	We sometimes overload notation and write $(c,c^*) \in M$ if $c \in C$ is a matched candidate center with mate $c^* \in C^*$.
\end{definition}

\begin{observation}
	\label{obs:double_counting_L}
	Since $|C| = |C^*| = k$, a counting argument tells us that $k - |M| \le 2 |L|$.
\end{observation}

We now define reassignment costs for candidate centers $c \in M \cup L \subseteq C$ and the notion of $\gamma$-good $OPT$ centers\footnote{For clarity, we define using $\gamma$. Later, we set $\gamma = \sqrt{\beta}$.}.
Informally, $OPT$ center $c^* \in C^*$ is $\gamma$-good if selecting a random point in $Q_{c^*}$ and removing a suitable candidate $c \in C$ reduces a ``sufficient'' fraction of the current clustering cost with a constant probability.

\begin{definition}[Reassignment costs]
	\hspace{0pt}\\
	If $(c,c^*) \in M$,
	\begin{multline*}
		\reassign(P,C,c)\\
		=  cost(P \setminus Q_{c^*}, C \setminus \{c\}) - cost(P \setminus Q_{c^*}, C)
	\end{multline*}
	If $c \in L$,
	\begin{multline*}
		\reassign(P,C,c) = cost(P, C \setminus \{c\}) - cost(P, C)
	\end{multline*}
\end{definition}

We will use the following lemma about reassignment costs, proven in Lemma 4 of \cite{lattanzi2019better}.

\begin{lemma}
	\label{prop:bound-reassignment-cost}
	For $c \in M \cup L$, with $P_c$ as the points assigned to $c$,
	\begin{multline*}
		\reassign(P,C,c) \le \frac{21}{100}\: cost(P_c, C) + 24\: cost(P_c, C^*)
	\end{multline*}
\end{lemma}

\begin{definition}[$M$-$\gamma$-good and $L$-$\gamma$-good]
	\label{def:good}
	\hspace{0pt}\\
	We say that $c^* \in C^* \cap M$ with mate $c \in C$ is $M$-$\gamma$-good if
	\begin{multline*}
		cost(Q_{c^*}, C) - \reassign(P,C,c) - 100 \cdot cost(Q_{c^*},c^*)\\
		> \frac{\gamma}{10^4k} \cdot cost(P,C).
	\end{multline*}
	We say that $c^* \in C^* \setminus M$ is $L$-$\gamma$-good if
	\begin{multline*}
		cost(Q_{c^*}, C) - \min_{c\in L}\reassign(P,C,c)\\
		- 100 \cdot cost(Q_{c^*},c^*) > \frac{\gamma}{10^4k} \cdot cost(P,C).
	\end{multline*}
\end{definition}

\begin{claim}
	\label{prop:Good_Prob_To_Sample_Center}
	Let $Q$ be a $M$-$\gamma$-good or $L$-$\gamma$-good cluster and we $D^2$-sample a point $q \in Q$.
	Then, with probability at least $2 / 5$, we have $cost(Q, q) \le 100 \cdot cost(Q, \mu_Q)$.
\end{claim}
\begin{proof}
	Let $C$ denotes the current set of candidate centers.
	Suppose cluster $Q$ is $M$-$\gamma$-good.
	Then,
	\begin{multline*}
		cost(Q,C) > \reassign(P,C,c) + 100 \cdot cost(Q,\mu_Q)\\
		+ \frac{\gamma}{10^4k} \cdot cost(P,C)
		\geq 100 \cdot cost(Q,\mu_Q)
	\end{multline*}
	By \cref{lem:sampled_point_is_from_unsettled_cluster}, we have $cost(Q,q) \leq 100 \cdot cost(Q,\mu_Q)$ with probability at least $1 - 6 / \sqrt{100} = 2 / 5$.
	The same argument holds when $Q$ is an $L$-$\gamma$-good cluster by applying the definition of $L$-$\gamma$-good instead.
\end{proof}

Conditioned on our main structural insight (\cref{lem:quality_deteriorates_slowly}), we sample a point from an $M$-$\sqrt\beta$-good or $L$-$\sqrt\beta$-good cluster $Q$ with constant probability, for every legal $\beta \in \mathcal{B}$ such that $4 \le \beta \le \alpha^{2/3}$ and $\alpha \ge 10^9$.
When this happens, the sampled point $s$ satisfies $cost(Q,s) \leq 100 \cdot cost(Q,\mu_Q)$ with constant probability.
In that case, the definition of $M$-$\sqrt\beta$-good and $L$-$\sqrt\beta$-good implies the existence of a candidate $t \in C$ with $cost(P, C \setminus \{t\} \cup \{s\}) \le (1 - \frac{\sqrt{\beta}}{10^4k}) \cdot cost(P,C)$, so \localsearch{} makes progress.

Similar to the analysis of Lattanzi and Sohler \yrcite{lattanzi2019better}, we partition the space of possible events into whether $\sum_{(c,c^*) \in M} cost(Q_{c^*}, C) \ge cost(P,C)/2$, or not.
In each case, we argue that the probability of sampling a point contained in a $M$-$\sqrt\beta$-good or $L$-$\sqrt\beta$-good cluster happens with a positive constant probability for a suitable legal $\beta$.
We first refine Lemma 5 of \cite{lattanzi2019better}.

\begin{lemma}
	\label{prop:M-good_are_big_fraction}
	Suppose $2 \cdot \sum_{(c,c^*) \in M} cost(Q_{c^*}, C) \geq cost(P,C) = \alpha \cdot cost(P,C^*)$ for $\alpha \ge 10^9$.
	Let $4 \le \beta \le \alpha^{2/3}$ be arbitrary.
	If there are at most $(30k) / \sqrt\beta$ clusters that are $\beta$-unsettled, then
	\[
		\sum_{{\substack{(c,c^*) \in M,\\ \text{$c^* \in M$-$\sqrt\beta$-good}}}} cost(Q_{c^*}, C) \ge \frac{1}{500} \cdot cost(P,C).
	\]
\end{lemma}

\cref{prop:M-good_are_big_fraction} tells us that if points in $M$ have sufficiently large probability mass, then points in $M$-$\sqrt\beta$-good clusters hold a constant fraction of the total probability mass.

\begin{proof}
	We show that the probability mass is large by upper bounding the probability mass on its negation.
	To do this, we partition the summation of $c^* \not\in M$-$\sqrt\beta$-good into $\beta$-settled and $\beta$-unsettled.
	We denote
	\begin{align*}
		\mathcal{A} & = \left\{(c,c^*) \in M, \text{$c^* \not\in M$-$\sqrt\beta$-good}, \text{$c^*$ is $\beta$-settled} \right\}   \\
		\mathcal{B} & = \left\{(c,c^*) \in M, \text{$c^* \not\in M$-$\sqrt\beta$-good}, \text{$c^*$ is $\beta$-unsettled} \right\}
	\end{align*}
	From \cref{prop:beta-settled}, we know that $C$ pays no more than $ (\beta +1) \cdot cost(P,C^*)$ for all $\beta$-settled clusters.
	So,
	\begin{multline*}
		\sum_{\mathcal{A}} cost(Q_{c^*}, C)
		\leq (\beta+1) \cdot cost(P,C^*)\\
		\leq (\alpha^{2/3}+1) \cdot cost(P,C^*)
		\leq \frac{2\alpha^{2/3}}{\alpha} \cdot cost(P,C)\\
		\leq \frac{1}{500} \cdot cost(P,C)
	\end{multline*}

	To bound $\sum_{\mathcal{B}} cost(Q_{c^*}, C)$, recall that $P$ is the set of \emph{all} points and $Q_{c^*} \subseteq P$ for any $c^* \in C^*$.
	\begin{align*}
		    & \sum_{\mathcal{B}} cost(Q_{c^*}, C)                                                     \\
		\le & \sum_{\mathcal{B}} \Big( \reassign(P,C,c) + 100 \cdot cost(Q_{c^*},c^*)                 \\
		    & + \frac{\sqrt\beta}{10^4k} \cdot cost(P,C) \Big)                           &  & (\star) \\
		\le & \left( \sum_{\mathcal{B}} \reassign(P,C,c) \right) + 100 \cdot cost(P,C^*)              \\
		    & + \frac{30}{10^4} \cdot cost(P,C)                                          &  & (\dag)  \\
		\le & \frac{21}{100} \cdot cost(P,C) + 24 \cdot cost(P,C^*)                                   \\
		    & + 100 \cdot cost(P,C^*) + \frac{30}{10^4} \cdot cost(P,C)                  &  & (\ddag) \\
		\le & \frac{250}{1000} \cdot cost(P,C)                                           &  & (\ast)
	\end{align*}

	\paragraph{(Legend)}
	$(\star)$: \cref{def:good};
	$(\dag)$: because there are at most $\frac{30 k}{\sqrt\beta}$ clusters that are $\beta$-unsettled;
	$(\ddag)$: \cref{prop:bound-reassignment-cost};
	$(\ast)$: $cost(P,C) \ge 10^9 \cdot cost(P,C^*)$

	Thus,
	\begin{multline*}
		\sum_{{\substack{(c,c^*) \in M,\\ \text{$c^* \in M$-$\sqrt\beta$-good}}}} cost(Q_{c^*}, C)\\
		\ge \left( \frac{1}{2} - \frac{1}{500} - \frac{250}{1000} \right) \cdot cost(P,C)
		\ge \frac{1}{500} \cdot cost(P,C)
	\end{multline*}
\end{proof}

Using the same structural insight on $\beta$-unsettled clusters, we now refine Lemma 7 of \cite{lattanzi2019better}.
Abusing notation, we use $C^* \setminus M$ to denote the set of optimal cluster centers which don't have a mate.
That is, the point $c \in C$ which $c^*$ is assigned to is assigned to has more than one center of $C^*$ assigned to it.

\begin{restatable}{lemma}{lemlgoodfraction}
	\label{prop:L-good_are_big_fraction}
	Suppose $2 \cdot \sum_{(c,c^*) \in M} cost(Q_{c^*}, C) < cost(P,C) = \alpha \cdot cost(P,C^*)$ for $\alpha \ge 10^9$.
	Let $4 \le \beta \le \alpha^{2/3}$  be arbitrary.
	If there are at most $(30k) / \sqrt\beta$ clusters that are $\beta$-unsettled, then
	\[
		\sum_{{\substack{c^* \in C^* \setminus M,\\ \text{$c^* \in L$-$\sqrt\beta$-good}}}} cost(Q_{c^*}, C) \ge \frac{1}{500} \cdot cost(P,C).
	\]
\end{restatable}

\cref{prop:L-good_are_big_fraction} tells us that if points in $C^* \setminus M$ have sufficiently large probability mass, then points in $L$-$\sqrt{\beta}$-good clusters hold a constant fraction of the total probability mass.

With \cref{prop:M-good_are_big_fraction} and \cref{prop:L-good_are_big_fraction}, we can now refine Lemma 3 from \cite{lattanzi2019better}.

\begin{restatable}{lemma}{refinedlemma}
	\label{lem:lemma3}
	Suppose we have a clustering $C$ with $cost(P,C) = \alpha \cdot cost(P, C^*)$ for some $\alpha \ge 10^9$.
	Assume that for each legal $\beta$, where $\beta \leq \alpha^{2/3}$, there are at most $(30k) / \sqrt\beta$ clusters that are $\beta$-unsettled.
	If we update $C$ to $C'$ in one \localsearch{} iteration, we have with probability at least $1/2000$:
	\[
		cost(P,C') \leq \left( 1 - \frac{\min \{ \sqrt[3]{\alpha},\, k^{0.15} \}}{2 \cdot 10^4 k} \right) \cdot cost(P,C)
	\]
\end{restatable}

\begin{proof}
	Pick a legal $\beta \in \mathcal{B}$ such that $\frac{1}{2} \min \left\{ k^{0.3}, \alpha^{2/3} \right\} \leq \beta < \min \left\{ k^{0.3}, \alpha^{2/3} \right\}$.
	We define $M$ and $L$ candidate centers as in \cref{def:HL} and consider the following two cases separately:
	\begin{enumerate}
		\item $\sum_{(c,c^*) \in M} cost(Q_{c^*},C) \geq \frac{1}{2} \cdot cost(P,C)$
		\item $\sum_{(c,c^*) \in M} cost(Q_{c^*},C) < \frac{1}{2} \cdot cost(P,C)$
	\end{enumerate}

	Let $q$ be the $D^2$-sampled point and $c \in C$ be some current candidate center, which we will define later in each case.
	In both cases (1) and (2), we will show that the pair of points $(q,c)$ will fulfill the condition $cost(P, C \cup \{q\} \setminus \{c\}) \le \left( 1 - \sqrt\beta / \left( 10^4\cdot k \right) \right) \cdot cost(P,C)$ with some constant probability.
	The claim follows since the algorithm takes the $c \in C$ that decreases the cost the most and swaps it with the $D^2$-sampled point $q$.

	\paragraph{Case (1):} $\sum_{(c,c^*) \in M} cost(Q_{c^*},C) \geq \frac{1}{2} \cdot cost(P,C)$

	\Cref{prop:M-good_are_big_fraction} tells us that we sample from a $M$-$\sqrt\beta$-good cluster with probability at least $1 / 500$. Denote this cluster as $Q_{c^*}$.
	Then, by \cref{prop:Good_Prob_To_Sample_Center}, the $D^2$-sampled point $q$ satisfies $cost(Q_{c^*}, q) \le 100 \cdot cost(Q_{c^*}, \mu_{Q_{c^*}})$ with probability at least $2 / 5$.
	Jointly, with probability at least $2 / 2500$, we $D^2$-sampled a ``good'' point $q \in Q_{c^*}$ where $(c,c^*) \in M$ and $c^* \in M$-$\sqrt\beta$-good, so
	\begin{align*}
		     & \; cost(P, C \cup \{q\} \setminus \{c\})                                                           \\
		=    & \; cost(P,C) - \left( cost(P,C) - cost(P, C \cup \{q\} \setminus \{c\}) \right)                    \\
		\leq & \; cost(P,C) - \Big((cost(P \setminus Q_{c^*},C) + cost(Q_{c^*},C))                                \\
		     & \; - (cost(P \setminus Q_{c^*}, C \setminus \{c\}) + cost(Q_{c^*}, q)) \Big)                       \\
		=    & \; cost(P,C) - \Big(cost(Q_{c^*},C) - (cost(P \setminus Q_{c^*}, C \setminus \{c\})                \\
		     & \; - cost(P \setminus Q_{c^*},C)) - cost(Q_{c^*}, q) \Big)                                         \\
		\le  & \; cost(P,C) - \Big( cost(Q_{c^*}, C) - \reassign(P,C,c)                                           \\
		     & \; - 100 \cdot cost(Q_{c^*}, \mu_{Q_{c^*}}) \Big)                                                  \\
		\le  & \; cost(P,C) - \frac{\sqrt\beta}{10^4\cdot k}\cdot cost(P,C)                                       \\
		=    & \; \left( 1 - \frac{\sqrt\beta}{10^4\cdot k} \right) \cdot cost(P,C)                               \\
		\le  & \; \left( 1 - \frac{\min\{\sqrt[3]{\alpha},k^{0.15}\}}{2\cdot 10^4\cdot k} \right) \cdot cost(P,C)
	\end{align*}

	\paragraph{Case (2):} $\sum_{(c,c^*) \in M} cost(Q_{c^*},C) < \frac{1}{2} \cdot cost(P,C)$

	This is the same as Case (1), but we use \Cref{prop:L-good_are_big_fraction} instead of \Cref{prop:M-good_are_big_fraction}.
\end{proof}

From this lemma, we can conclude that if the current approximation factor is very high, we drastically decrease it within just a few steps.
In particular, we can show that if we start with an approximation guarantee that is no worse than $\exp(k^{0.1})$, we can decrease it to just $O(1)$ within $\eps k$ steps, with probability $1 - \exp(-\Omega(k^{0.1}))$.
By Markov's inequality, we know that the probability of having an approximation guarantee that is worse than $\exp(k^{0.1})$ is at most $\exp(-\Omega(k^{0.1}))$.
Our main theorem\footnote{The cube-root of $\alpha$ in \Cref{lem:lemma3} is precisely why we obtain an approximation factor of $\O(1 / \eps^3)$ after $\eps k$ \localsearch{} steps, with high probability in $k$.} now follows:

\mainthm*

\subsection{Concluding remarks}

\paragraph{Expectation versus high probability}

An approximation guarantee in expectation only implies (via Markov inequality) that with a constant probability we get a constant approximation.
So, our result is stronger as we get a constant approximation of the optimum cost with a probability of at least $1 - \exp(-\Omega(k^{0.1}))$.
To recover a guarantee in expectation, we can run the algorithm twice\footnote{It is not unusual to run \kmp{} multiple times in practice. e.g. See documentation of \texttt{sklearn.cluster.KMeans}.}:
Let $C_1$ be the solution obtained by running \kmp{} plus \localsearch{}, let $C_2$ be the output of another independent run of \kmp{}, and let $\mathcal{E}$ be the event that \localsearch{} outputs an $\O(1)$-approximation.
Then, the expected cost of $\min\{cost(C_1), cost(C_2)\}$ is
\begin{align*}
	     & \; \E[\min\{cost(C_1), cost(C_2)\}]                                          \\
	\leq & \; \Pr[\mathcal{E}] \cdot \O(1) + (1 - \Pr[\mathcal{E}]) \cdot \E[cost(C_2)] \\
	\leq & \;\left(1 - \exp(-\Omega(k^{0.1})) \right) \cdot \O(1)                       \\
	     & \; + \exp(-\Omega(k^{0.1})) \cdot \O(\log k)                                 \\
	\in  & \; \O(1)
\end{align*}

\paragraph{Running time}
\label{para:runningtime}
On a $d$-dimensional data set consisting of $n$ data points, a naive implementation of \kmp{} has time complexity $\O(dnk^2)$ and space complexity $\O(dn)$.
This running time can be improved to $\O(dnk)$ if each data point tracks its distance to the current closest candidate center in $C$.
This is because $D^2$-sampling and subsequent updating this data structure can be done in $\O(dn)$ time for each iteration of \kmp{}.

\localsearch{} can be implemented in a similar manner where each data point remembers its distance to the closest two candidate centers.
Lattanzi and Sohler \cite{lattanzi2019better} argue that if \localsearch{} deletes clusters with an average size of $\O(n/k)$, then an iteration of \localsearch{} can be performed in an amortized running time of $\O(dn)$.

However, in the worst case, each iteration of \localsearch{} can still take $\O(dnk)$ time.
A way to provably improve the worst case complexity is to use more memory.
With $\O(dnk)$ space, each data point can store distances to all $k$ centers in a binary search tree.
Then, each step can be implemented in $\O(dn \log k)$ time, as updating a binary search tree requires $\O(\log k)$ time.

\section{Acknowledgements}

We are grateful to Zalan Borsos, Mohsen Ghaffari, and Andreas Krause for their help and discussing this problem with us.
In particular, we thank Mohsen Ghaffari for giving us feedback on previous versions of this paper.

\bibliography{ref}
\bibliographystyle{icml2020}

\newpage
\onecolumn
\appendix

\section{Missing proofs}
Here we collect missing proofs of the results left unproven in the main body.

\subsection{Concentration Inequalities}
We will use the standard Chernoff Bound:
\begin{theorem}[Chernoff bound]
	\label{lem:chernoff}
	Let $X_1, X_2, \ldots, X_n$ be independent binary random variables.
	Let $X = \sum_{i=1}^n X_i$ be their sum and $\mu = \E(X)$ be the sum's expected value.
	For any $\delta \in [0,1]$,
	$\P(X \leq (1-\delta)\mu) \leq \exp \left( - \delta^2 \mu / 2 \right)$ and $\P(X \geq (1+\delta)\mu) \leq \exp \left(- \delta^2 \mu / 3 \right)$.
	Additionally, for any $\delta \geq 1$, we have
	$\P(X \geq (1+\delta)\mu) \leq \exp \left(- \delta \mu / 3 \right)$.
\end{theorem}

\subsection{Proof of \Cref{lem:quality_deteriorates_slowly}}

We will use \Cref{cor:I_like_Sebastians_chocolate} together with a Chernoff Bound to prove  \Cref{prop:kmeans++_doesnt_screwup}, which we will then use to prove \Cref{lem:quality_deteriorates_slowly}.
\begin{proposition}
	\label{prop:kmeans++_doesnt_screwup}
	After running \kmp{} (for $k$ steps), let $C$ denote the set of candidate centers and let $\alpha = \frac{cost(P,C)}{OPT}\geq 1$ be the approximation factor.
	Then, with probability at least $1 - \exp(-\Omega(k^{0.1}))$, there are at most $\left( \frac{10 k}{\sqrt \beta} \right)$ $\beta$-unsettled clusters for any legal $\beta \leq \alpha^{2/3}$.
\end{proposition}
\begin{proof}
	Define $\alpha_i = \frac{cost(P,C_{i})}{OPT}$ as the approximation factor after the $i$-th step of \kmp{}.
	Fix an arbitrary legal $\beta$.
	Since $\alpha_1 \ge \alpha_2 \ge \dots \ge \alpha_k = \alpha$, $\beta \le \alpha^{2/3}$ implies $\beta \le \alpha_i^{2/3}$ for any $i$.
	Note that if $\beta > \alpha^{2/3}$, then the statement vacuously holds.

	Let $X_i$ be an indicator random variable which is $1$ if $\beta \leq \alpha_{i-1}^{2/3}$ and we \emph{do not} increase the number of $\beta$-settled clusters by one in the $i$-th iteration.
	Then, $\E[X_i|X_1, \ldots, X_{i-1}] \leq \frac{8}{\sqrt{\beta}}$ for any $X_1, \dots, X_{i-1}$.
	This is because if $\beta \le \alpha_{i-1}^{2/3}$, then \cref{cor:I_like_Sebastians_chocolate} tells us that in the $i$-th iteration, with probability at least $1 - \frac{8}{\sqrt\beta}$, the new sampled point is from a $\beta$-unsettled cluster $Q$ \emph{and} makes it $\beta$-settled, so the number of $\beta$-settled clusters increases by one.

	Define random variables $X = X_1 + \dots + X_k$ and $X' = X'_1 + \dots + X'_k$, where each $X'_i$ is an independent Bernoulli random variable with success probability $\frac{8}{\sqrt\beta}$.
	We see that $\E[X'] = \frac{8k}{\sqrt\beta}$ and $X$ is stochastically dominated by $X'$.
	By \cref{lem:chernoff},
	\[
		\P \left( X \ge \frac{10 k}{\sqrt\beta} \right)
		\le \P \left( X' \ge \frac{10 k}{\sqrt\beta} \right)
		= \P \left( X' \ge \frac{5}{4} \cdot \E[X'] \right)
		\le \e^{-\frac{\E[X']^2}{3 \cdot 16}}
		\le \e^{-\Theta(\frac{k}{\sqrt{\beta}})}
		\le \e^{-\Theta \left( k^{0.85} \right)}
	\]
	The last inequality holds because $\beta \le k^{0.3}$ for any legal $\beta \in \mathcal{B}$.
	Since we start with $k$ $\beta$-unsettled clusters, if $X \le \frac{10k}{\sqrt\beta}$ and $\beta \leq \alpha^{2/3}$, then the number of $\beta$-unsettled clusters at the end is at most $k - (k - \frac{10k}{\sqrt\beta}) = \frac{10k}{\sqrt\beta}$.
	To complete the proof, we union bound over all $\O(\log k)$ possible values for legal $\beta \in \mathcal{B}$.
\end{proof}

\lemqualitygood*
\begin{proof}
	Define $\alpha_0$ as the approximation factor after \kmp{} and $\alpha_i$ as the approximation factor after running the local search for additional $i$ steps.
	Fix an arbitrary legal $\beta$.
	Since $\alpha_0 \ge \alpha_1 \ge \dots \ge \alpha_\ell = \alpha$, $\beta \le \alpha^{2/3}$ implies $\beta \le \alpha_i^{2/3}$ for any $i$.
	Note that if $\beta > \alpha^{2/3}$, then the statement vacuously holds.

	Let $X_i$ be an indicator random variable which is $1$ if $\beta \leq \alpha_{i-1}^{2/3}$ and the number of $\beta$-unsettled clusters increases in the $i$-th iteration of \localsearch{}.
	Then, $\E[X_i \mid X_1, \dots, X_{i-1}] \le \frac{8}{\sqrt\beta}$.
	This is because if $\beta \le \alpha_{i-1}^{2/3}$ in the $i$-th iteration, then \cref{cor:I_like_Sebastians_chocolate} tells us that, with probability at least $1 - \frac{8}{\sqrt{\beta}}$, the new sampled point $p_i$ is from a $\beta$-unsettled cluster $Q$ \emph{and} adding $p_i$ to $C$ would make $Q$ $\beta$-settled.
	By definition of $\beta$-settled, adding or removing a single point from $C$ can only decrease or increase the number of $\beta$-unsettled clusters by at most one.
	Thus, if \localsearch{} decides to swap an existing point in $C$ for $p_i$, the number of $\beta$-unsettled clusters does not increase.

	Define random variables $X = X_1 + \dots + X_l$ and $X' = X'_1 + \dots + X'_l$, where each $X'_i$ is an independent Bernoulli random variable with success probability $\frac{8}{\sqrt\beta}$.
	We see that $\E[X'] = \frac{8l}{\sqrt\beta}$ and $X$ is stochastically dominated by $X'$.
	By \cref{lem:chernoff},
	\[
		\P \left( X \ge \frac{20k}{\sqrt\beta} \right)
		\le \P \left( X' \ge \frac{20k}{\sqrt\beta} \right)
		= \P \left( X' \ge \frac{5k}{2l} \cdot \E[X'] \right)
		\le \e^{-\frac{\E[X'] \frac{3k}{2\ell}}{3}}
		\le \e^{-\Theta(\frac{k}{\sqrt{\beta}})}
		\le \e^{-\Theta \left( k^{0.85} \right)}
	\]
	The last inequality is because $\beta \le k^{0.3}$ for any legal $\beta \in \mathcal{B}$.
	\Cref{prop:kmeans++_doesnt_screwup} tells us at the start of \localsearch, with probability at least $1 - e^{-\Omega(k^{0.1})}$, there are at most $\frac{10k}{\sqrt\beta}$ $\beta$-unsettled clusters.
	If $X \le \frac{20k}{\sqrt\beta}$ and $\beta \leq \alpha^{2/3}$, then the number of $\beta$-unsettled clusters after $l$ \localsearch{} steps is at most $\frac{10k}{\sqrt\beta} + \frac{20k}{\sqrt\beta}   \leq \frac{30k}{\sqrt\beta}$.
	To complete the proof, we union bound over all $\O(\log k)$ possible values for legal $\beta \in \mathcal{B}$.
\end{proof}

\subsection{Proof of \Cref{prop:L-good_are_big_fraction}}

We proof \cref{prop:L-good_are_big_fraction}, which is similar to the proof of \cref{prop:M-good_are_big_fraction}.

\lemlgoodfraction*

We abuse notation and denote with $C^* \setminus M$ the set of optimal cluster centers which don't have a mate. That is, the point $c \in C$ which $c^*$ is assigned to is assigned to more than one center of $C^*$. Informally, the proposition states that if points in $C^* \setminus M$ have sufficiently large probability mass, then the probability mass on $L$-$\sqrt{\beta}$-good clusters is a constant fraction of the total probability mass.

\begin{proof}

	We show that the probability mass is large by upper bounding the probability mass on its negation.
	To do this, we partition the summation of $c^* \not\in L$-$\sqrt\beta$-good into $\beta$-settled and $\beta$-unsettled:
	\[
		\sum_{{\substack{c^* \in C^* \setminus M,\\ \text{$c^* \not\in L$-$\sqrt\beta$-good}}}} cost(Q_{c^*}, C)
		= \sum_{{\substack{c^* \in C^* \setminus M,\\ \text{$c^* \not\in L$-$\sqrt\beta$-good},\\ \text{$c^*$ is $\beta$-settled}}}} cost(Q_{c^*}, C) + \sum_{{\substack{c^* \in C^* \setminus M,\\ \text{$c^* \not\in L$-$\sqrt\beta$-good},\\ \text{$c^*$ is $\beta$-unsettled}}}} cost(Q_{c^*}, C)
	\]
	From \cref{prop:beta-settled}, we know that $C$ pays no more than $ (\beta +1) \cdot cost(P,C^*)$ for all $\beta$-settled clusters.
	So,
	\[
		\sum_{{\substack{c^* \in C^* \setminus M,\\ \text{$c^* \not\in L$-$\sqrt\beta$-good},\\ \text{$c^*$ is $\beta$-settled}}}} cost(Q_{c^*}, C)
		\leq (\beta+1) \cdot cost(P,C^*)
		\leq (\alpha^{2/3}+1) \cdot cost(P,C^*)
		\leq \frac{2\alpha^{2/3}}{\alpha} \cdot cost(P,C)
		\leq \frac{1}{500} \cdot cost(P,C)
	\]
	It remains to bound $\sum_{{\substack{c^* \in C^* \setminus M,\\ \text{$c^* \not\in L$-$\sqrt\beta$-good},\\ \text{$c^*$ is $\beta$-unsettled}}}} cost(Q_{c^*}, C)$.
	Recall that $P$ is the set of \emph{all} points and $Q_{c^*} \subseteq P$ for any $c^* \in C^*$.
	For $c^* \in C^* \setminus M$, if $c^* \not\in L$-$\sqrt\beta$-good, then for \emph{any} $c \in L$, $cost(Q_{c^*}, C) \le \reassign(P,C,c) + 100 \cdot cost(Q_{c^*},c^*) + \frac{\sqrt\beta}{10^4k} \cdot cost(P,C)$.
	\begin{align*}
		    & \sum_{{\substack{c^* \in C^* \setminus M,                                                                                                                                          \\ \text{$c^* \not\in L$-$\sqrt\beta$-good},\\ \text{$c^*$ is $\beta$-unsettled}}}} cost(Q_{c^*}, C)\\
		\le & \sum_{{\substack{c^* \in C^* \setminus M,                                                                                                                                          \\ \text{$c^* \not\in L$-$\sqrt\beta$-good},\\ \text{$c^*$ is $\beta$-unsettled}}}} \left( \min_{c\in L}\reassign(P,C,c) + 100 \cdot cost(Q_{c^*},c^*) + \frac{\sqrt\beta}{10^4k} \cdot cost(P,C) \right) && \text{\cref{def:good}}\\
		\le & \left( \sum_{{\substack{c^* \in C^* \setminus M,                                                                                                                                   \\ \text{$c^* \not\in L$-$\sqrt\beta$-good},\\ \text{$c^*$ is $\beta$-unsettled}}}} \min_{c\in L} \reassign(P,C,c) \right) + 100 \cdot cost(P,C^*) + \frac{30}{10^4} \cdot cost(P,C) && \text{$\leq \frac{30 k}{\sqrt\beta}$ $\beta$-unsettled}\\
		\le & (k- |M|) \min_{c\in L} \reassign(P,C,c) + 100 \cdot cost(P,C^*) + \frac{30}{10^4} \cdot cost(P,C)                              &  & \text{Sum over $\le|C^*\setminus M|$ elements} \\
		\le & 2|L| \min_{c\in L} \reassign(P,C,c)  + 100 \cdot cost(P,C^*) + \frac{30}{10^4} \cdot cost(P,C)                                 &  & \text{\cref{obs:double_counting_L}}            \\
		\le & 2\sum_{c \in L} \reassign(P,C,c)  + 100 \cdot cost(P,C^*) + \frac{30}{10^4} \cdot cost(P,C)                                    &  &                                                \\
		\le & 2\left( \frac{21}{100} \cdot cost(P,C) + 24 \cdot cost(P,C^*)\right) + 100 \cdot cost(P,C^*) + \frac{30}{10^4} \cdot cost(P,C) &  &
		\text{\cref{prop:bound-reassignment-cost}}                                                                                                                                               \\
		\le & \frac{450}{1000} \cdot cost(P,C)                                                                                               &  & cost(P,C) \ge 10^9 \cdot cost(P,C^*)
	\end{align*}
	Thus, $\sum_{{\substack{c^* \in C^* \setminus M,\\ \text{$c^* \in L$-$\beta$-good}}}} cost(Q_{c^*}, C) \ge (\frac{1}{2} - \frac{1}{500} - \frac{450}{1000}) \cdot cost(P,C) \ge \frac{1}{500} \cdot cost(P,C)$.
\end{proof}

\subsection{Proof of \Cref{thm:main}}
Before we prove \Cref{thm:main}, we will introduce the notion of a \emph{successful} iteration of local search, and argue under which conditions we can give guarantees on the probability that an iteration is successful.
In \Cref{prop:successes}, we give a lower bound on the number of successful rounds that we expect to see, and in \Cref{prop:phases1} and \Cref{prop:phases2} we show that after enough successful rounds, a significant decrease in cost is achieved.
This finally enables us to prove \Cref{thm:main}.

For our analysis, we will require that the following two events hold before every step of \localsearch{}.
\begin{enumerate}
	\item [(I)] If we start with an approximation factor of $\alpha$, then for every legal $\beta \leq \alpha^{2/3}$, there are at most $\frac{30k}{\sqrt{\beta}}$ $\beta$-unsettled clusters.
	      Assuming we perform $\ell \leq k$ steps of local search, we can assume that this is true by \Cref{lem:quality_deteriorates_slowly}, and a union bound over all of the at most $k$ steps, with probability at least $1 - \exp(-\Omega(k^{0.1}))$.
	\item [(II)] We will also assume that the approximation factor after the execution of \kmp{} (which can only improve) is at most $\exp(k^{0.1})$.
	      As the expected cost is $\O(\log k)$, this occurs with probability at least $1 - \exp(- \Omega(k^{0.1}))$, by a simple application of Markov's inequality.
\end{enumerate}
Both statements (I) and (II) jointly hold with probability at least $1 - \exp(-\Omega(k^{0.1}))$.

Letting $\alpha_{i}$ denote the approximation factor after the $i$-th local search step, we define a \emph{successful} local search step as follows:
\begin{definition}
	The $i$-th local search step is \emph{successful} if either of the following holds:
	(A) $\alpha_{i-1} \leq 10^9$, or
	(B) $\alpha_i \leq \left( 1 - \frac{\min \{ \sqrt[3]{\alpha_{i-1}},\, k^{0.15} \}}{2 \cdot 10^4 k} \right) \cdot \alpha_{i-1}$
\end{definition}

Note, as we condition on (I) and (II), we cannot directly apply \Cref{lem:lemma3} to show that an iteration is successful with probability at least $1/2000$.
However, the following is still true:
\begin{observation}
	\label{obs:lemma3}
	As $1 - \exp(-\Omega(k^{0.1})) \gg 1 - 1/4000$, the probability that an iteration of local search is successful, after conditioning on (I) and (II), is still at least $1/4000$.
\end{observation}
Using this observation, we can now state the following:
\begin{proposition}
	\label{prop:successes}
	Assume that we run $\ell \leq k$ local search steps and that conditions (I) and (II) hold. Then, with a probability of at least $1 - \exp(-\Omega(\ell))$, we will have at least $\ell/ 8000$ successes.
\end{proposition}
\begin{proof}
	Let $X_i$ denote the indicator variable for the event that the $i$-th local search step is a success. Note that $X_1,X_2, \ldots, X_\ell$ are not independent.
	However, it is easy to check that \cref{obs:lemma3} holds, even if we additionally condition on arbitrary values for $X_1, X_2, \ldots, X_{i-1}$, or more specifically:
	$$\E[X_i|X_1,X_2, \ldots, X_{i-1},(I),(II)] \geq \frac{1}{4000}$$
	Thus, the number of successes stochastically dominates the random variable $X' = X'_1 + \cdots + X'_\ell$, where the $X'_i$s are independent Bernoulli variables that are one with probability $1/4000$.
	By a Chernoff bound, we can thus conclude that within $\ell$ rounds, less than $\ell/8000$ rounds are successful, with probability at most $\exp(-\Omega(l))$.
\end{proof}

\begin{proposition}
	\label{prop:phases1}
	Assume that the conditions (I) and (II) are fulfilled.
	Then, after $N_0 = 2 \cdot 10^4 \cdot k^{0.95}$ successful rounds, we obtain a clustering which is no worse than $k^{0.45}$-approximate, assuming $k^{0.45} > 10^9$.
\end{proposition}
\begin{proof}
	As we condition on (II), the initial approximation factor is no worse than $\exp(k^{0.1})$.
	For the sake of contradiction, assume that the approximation factor after $N_0$ successes is strictly greater than $k^{0.45}$.
	This implies that in each of the first $N_0$ successful rounds, we improve the approximation by a factor of at least $(1 - \frac{ k^{0.15} }{2 \cdot 10^4 k})$.
	Thus, the approximation factor after $N_0$ successes is at most
	\begin{align*}
		\left(1 - \frac{k^{0.15}}{2\cdot10^4 k}\right)^{2\cdot10^4 \cdot k^{0.95}} \cdot \exp(k^{0.1})
		 & = \left(1 - \frac{k^{-0.85}}{2\cdot10^4} \right)^{2\cdot10^4 \cdot k^{0.95}} \cdot \exp(k^{0.1})           \\
		 & \leq \exp \left( \frac{-k^{-0.85}}{2\cdot10^4} \cdot 2\cdot10^4 \cdot k^{0.95} \right) \cdot \exp(k^{0.1}) \\
		 & \leq \exp(-k^{0.1} + k^{0.1}) \leq 1 \leq k^{0.45},
	\end{align*}
	a contradiction.
\end{proof}
\begin{proposition}
	\label{prop:phases2}
	Assume that conditions (I) and (II) are fulfilled.
	We define $\gamma_i := \frac{k^{0.45}}{2^{i}}$.
	Furthermore, for $i \geq 1$, let $N_i := \frac{2 \cdot 10^4 \cdot k}{\sqrt[3]{\gamma_i}}$.
	Then, for each $R \geq 0$, after $\sum_{i=0}^R N_i$ successes, we have a $\max \{ \gamma_R,10^9 \}$-approximation.
\end{proposition}
\begin{proof}
	We prove the statement by induction on $R$. For $R = 0$, the statement directly follows from \cref{prop:phases1}.
	Now, let $R > 0$ be arbitrary.
	We assume that the statement holds for $R$, and we show that this implies that the statement holds for $R+1$.
	For the sake of contradiction, assume that the statement does not hold for $R+1$, i.e., the approximation is strictly worse than $\max\{\gamma_{R+1}, 10^9\}$ after an additional $N_{R+1}$ successful rounds.
	In particular, this would mean that we never achieve a $10^9$-approximation.
	Thus, in each of the additional $N_{R+1}$ successful iterations, we would improve the solution by a factor of at least $(1 - \frac{ \sqrt[3]{\gamma_{R+1}} }{2 \cdot 10^4 k})$.
	As we started with an approximation factor no worse than $\gamma_R$, the approximation factor after $N_{R+1}$ successful rounds can be upper bounded by
	\begin{align*}
		\left(1 - \frac{\sqrt[3]{\gamma_{R+1}}}{2\cdot10^4 k}\right)^{\frac{2\cdot10^4 k}{\sqrt[3]{\gamma_{R+1}}}} \cdot \gamma_R
		 & \leq \exp \left( -\frac{\sqrt[3]{\gamma_{R+1}}}{2\cdot10^4 k} \cdot \frac{2\cdot10^4 k}{\sqrt[3]{\gamma_{R+1}}} \right) \cdot \gamma_R \\
		 & \leq e^{-1} \cdot \gamma_R < \gamma_{R+1},
	\end{align*}
	a contradiction.
\end{proof}

Finally, we can prove \cref{thm:main}.

\mainthm*
\begin{proof}[Proof of \cref{thm:main}]
	First, recall that we can assume that conditions (I) and (II) are fulfilled, which holds with probability $1 - \exp(-\Omega(k^{0.1}))$.
	Let $\alpha_0 \leq \exp(k^{0.1})$ be our approximation factor after the execution of \kmp{}.
	From \Cref{prop:phases2}, we know that after $\sum_{i=0}^R N_i$ successful iterations we have an approximation factor of at most $\gamma_R = \max \{ \frac{k^{0.45}}{2^R}, 10^9 \}$.
	Setting $R = \log_2(k^{0.45} \eps^3) - 3 \log_2(32 \cdot 10^8)$, we get that $\gamma_R \leq \frac{k^{0.45}}{2^R} \leq \frac{10^{30}}{\eps^3}$.
	For the number of successful iterations needed, we have:
	\begin{align*}
		\sum_{i=0}^{\log_2(k^{0.45} \eps^3) - 3 \log_2(32 \cdot 10^8)} N_i
		 & \leq 2 \cdot 10^4 \cdot k^{0.95} + \sum_{i=1}^{\log_2(k^{0.45} \eps^3) - 3 \log_2(32 \cdot 10^8)} \frac{2 \cdot 10^4 \cdot k}{\sqrt[3]{k^{0.45}/2^i}}                                                         \\
		 & \leq 2 \cdot 10^4 \cdot k^{0.95} + 2 \cdot 10^4 \cdot k \sum_{i=1}^{\log_2(k^{0.45} \eps^3) - 3 \log_2(32 \cdot 10^8)} \sqrt[3]{2^i / k^{0.45}}                                                               \\
		 & \leq 2 \cdot 10^4 \cdot k^{0.95} + 2 \cdot 10^4 \cdot k \cdot \frac{1}{1 - 1 / \sqrt[3]{2}} \cdot 2^{\left( \log_2(k^{0.45} \eps^3) - 3 \log_2(32 \cdot 10^8) \right) / 3} \cdot \frac{1}{\sqrt[3]{k^{0.45}}} \\
		 & \leq 2 \cdot 10^4 \cdot k^{0.95} + \frac{\eps k}{16 \cdot 10^3} \leq \frac{\eps k}{8000}.
	\end{align*}
	By \Cref{prop:successes}, we can conclude that within $\eps k$ steps of local search, at least $\eps k / 8000$ are successful with probability at least $1 - \exp(- \Omega(k^{0.1}))$, thus proving the theorem.
\end{proof}

\end{document}